\documentclass[letterpaper,10pt,conference]{ieeeconf}

\usepackage{cite}
\usepackage{graphicx}
\usepackage{tikz}
\usepackage{subcaption}
\usepackage[cmex10]{amsmath}

\usepackage{algpseudocode}
\usepackage{algorithmicx} 
\usepackage{array}
\usepackage{multirow}

\usepackage{epsfig}
\usepackage{amsfonts,amssymb,multirow,bigstrut,booktabs,ctable,latexsym}

\usepackage[mathscr]{eucal}
\usepackage{verbatim}

\usepackage{amsthm}
\usepackage{algorithm}

\usepackage{stmaryrd}

\IEEEoverridecommandlockouts  
\begin{document}

\newtheorem{definition}{Definition}
\newtheorem{lemma}{Lemma}
\newtheorem{corollary}{Corollary}
\newtheorem{theorem}{Theorem}
\newtheorem{example}{Example}
\newtheorem{proposition}{Proposition}
\newtheorem{remark}{Remark}
\newtheorem{assumption}{Assumption}
\newtheorem{corrolary}{Corrolary}
\newtheorem{property}{Property}
\newtheorem{ex}{EX}
\newtheorem{problem}{Problem}
\newcommand{\argmin}{\arg\!\min}
\newcommand{\argmax}{\arg\!\max}
\newcommand{\st}{\text{s.t.}}
\newcommand \dd[1]  { \,\textrm d{#1}  }

\title{\bf \Large A Compositional Resilience Index for Computationally Efficient Safety Analysis of Interconnected Systems}

\author{Luyao Niu$^{1}$, Abdullah Al Maruf$^{1}$, Andrew Clark$^2$, J. Sukarno Mertoguno$^3$, and Radha Poovendran$^1$ %
\thanks{$^1$Luyao Niu, Abdullah Al Maruf, and Radha Poovendran are with the Network Security Lab, Department of Electrical and Computer Engineering,
University of Washington, Seattle, WA 98195-2500
        {\tt\small \{luyaoniu,maruf3e,rp3\}@uw.edu}}
\thanks{$^2$Andrew Clark is with the Electrical and Systems Engineering Department, McKelvey School of Engineering, Washington
University in St. Louis, St. Louis, MO 63130
        {\tt\small andrewclark@wustl.edu}}%
\thanks{$^3$J. Sukarno Mertoguno is with the School of Cybersecurity and
Privacy, Georgia Institute of Technology, Atlanta, GA 30332
{\tt\small karno@gatech.edu}}
}

\maketitle

\begin{abstract}
Interconnected systems such as power systems and chemical processes are often required to satisfy safety properties in the presence of faults and attacks.
Verifying safety of these systems, however, is computationally challenging due to nonlinear dynamics, high dimensionality, and combinatorial number of possible faults and attacks that can be incurred by the subsystems interconnected within the network.
In this paper, we develop a compositional resilience index to verify safety properties of interconnected systems under faults and attacks.
The resilience index is a tuple serving the following two purposes.
First, it quantifies how a safety property is impacted when a subsystem is compromised by faults and attacks.
Second, the resilience index characterizes the needed behavior of a subsystem during normal operations to ensure  safety violations will not occur when future adverse events occur.
We develop a set of sufficient conditions on the dynamics of each subsystem to satisfy its safety constraint, and leverage these conditions to formulate an optimization program to compute the resilience index.
When multiple subsystems are interconnected and their resilience indices are given, we show that the safety constraints of the interconnected system can be efficiently verified by solving a system of linear inequalities.
We demonstrate our developed resilience index using a numerical case study on chemical reactors connected in series.
\end{abstract}


\section{Introduction}

Safety-critical interconnected systems are widely seen in real-world applications such as power systems \cite{sullivan2017cyber} and chemical processes \cite{el2005fault}.
Safety violations can lead to significant economic losses and severe damage to the system and/or human operators engaged with the system \cite{sullivan2017cyber,Jeep, fan2020fast,koscher2010experimental,chen2020guaranteed}.
Therefore, it is of critical importance to verify safety properties for such large-scale or even societal-scale systems.

One approach to verify safety is to use reachability analysis.
Computing reachable sets for nonlinear systems is known to be undecidable \cite{fijalkow2019decidability}.
Alternatively, solutions to safety verification by ensuring forward invariance of safety sets \cite{prajna2007framework,niu2022verifying,clark2021verification,manna2012temporal} or approximating reachable sets \cite{mitchell2005time,bertsekas1971minimax,fisac2018general} have been developed.
However, these approaches do not scale to interconnected systems of high dimensions.
Large-scale systems such as power systems generally consist of multiple interconnected subsystems, motivating the development of compositional approaches \cite{sloth2012compositional,lyu2022small,coogan2014dissipativity,anand2021small}.
These approaches decompose the safety verification problem into a set of problems of smaller scales formulated on the subsystems, and thus are more tractable. 


The approaches in \cite{sloth2012compositional,lyu2022small,coogan2014dissipativity,anand2021small} assume that the systems are operated under benign environments, making the verified safety properties invalid for systems under faults and attacks. 
For interconnected systems, an error from one faulty or compromised subsystem could propagate and accumulate through interconnections and impact the safety of other subsystems.
A na\"ive approach to safety verification for interconnected systems operated under adversarial environments is to enumerate all possible faulty or compromised subsystems, and perform safety analysis.
However, the number of possible faults or attacks that can be incurred by the interconnected system is combinatorial.
At present, scalable safety verification of large-scale interconnected systems under faults and attacks has been less studied.


In this paper, we develop a compositional safety verification approach for large-scale interconnected systems whose subsystems can be faulty or compromised by attacks.
Each subsystem is subject to a safety constraint.
We derive a set of conditions on the dynamics of a subsystem to guarantee its safety.
We parameterize these conditions using a tuple of real numbers, termed resilience index.
Our resilience index defines the amount of time that the system can safely remain in a faulty state, the amount of time required to recover from faults, and constraints on the system dynamics that must be satisfied during faulty as well as normal operation.
The resilience index allows us to convert the problem of safety verification of large-scale interconnected systems to a set of algebraic computations, and thus makes safety verification feasible for large-scale systems.
To summarize, this paper makes the following contributions.
\begin{itemize}
    \item 
    We formulate a resilient index for a subsystem that experiences faults or attacks.
    We prove safety guarantees for a subsystem based on the resilience index. We develop a sum-of-squares optimization to compute the resilience index for a subsystem.
    \item 
    We derive a system of linear inequalities to quantify how the resilience index of a subsystem changes due to interconnections.
    Using the derived linear inequalities, we develop the conditions on the interconnections so that all subsystems are safe under faults and attacks.
    \item We demonstrate the proposed resilience indices and their usage for safety analysis by using a numerical case study on chemical process.
\end{itemize}

The rest of this paper is organized as follows. Section \ref{sec:related} presents related work. 
Section \ref{sec:formulation} describes the problem formulation.
Section \ref{sec:RI} develops the compositional resilience index for each subsystem.
In Section \ref{sec:interconnection}, we derive the set of inequalities to compute resilience indices after interconnection.
Section \ref{sec:case study} demonstrates the proposed approach using a numerical case study.
We conclude the paper in Section \ref{sec:conclusion}.

\section{Related Work}\label{sec:related}

Safety verification \cite{clarke1997model,manna2012temporal,prajna2007framework} and safety-critical control \cite{ames2016control,xu2018constrained,fan2020fast,breeden2021guaranteed} haven been investigated for systems operated in benign environments. 
To mitigate faults and attacks against safety-critical systems, various techniques have been developed. 
Attack detection and secure state estimation under attacks have been studied in \cite{fawzi2014secure,pajic2014robustness}.
Fault-tolerant and resilient control schemes \cite{sha2001using,abdi2018guaranteed,mertoguno2019physics,niu2022verifying} have been proposed to withstand the attacks and guarantee system safety.
For interconnected systems consisting of multiple subsystems, the systems are of high dimensionality and the attack surface grows as interconnected systems involving more subsystems, making safety verification computationally expensive.

Compositional approaches have been adopted for safety verification of interconnected systems deployed in the absence of faults or attacks \cite{sloth2012compositional,lyu2022small,coogan2014dissipativity,anand2021small}.
These approaches have utilized techniques including barrier certificates \cite{sloth2012compositional,anand2021small}, small-gain theorem \cite{lyu2022small}, and dissipativity property \cite{coogan2014dissipativity}.
When the system is operated under faulty or adversarial environments, these approaches become less effective.

The authors of \cite{yang2015fault,yang2020exponential} re-configured the control laws of each subsystem and interconnection topology to guarantee safety.
Such approach is computationally expensive when re-configuring the network topology and control laws.
Furthermore, in applications such as power systems, re-designing interconnection topology is less desired or even impractical.
In this paper, we develop a compositional resilience index and prove that the safety of each subsystem interconnected within a network can be analyzed by solving a system of linear inequalities derived using resilience index. 
Our developed approach does not require re-configuring the network topology or control laws, and hence is more computationally efficient.
In \cite{al2022compositional}, the authors decomposed the dynamics of each subsystem into intrinsic and coupled terms, where the former term is independent of the other subsystems and the latter one depends on interconnections.
A resilience index was defined for each term of a compromised subsystem by bounding how fast a subsystem approaches the boundary of safety set.
Such resilience indices were computed by sum-of-squares optimization, and allowed the synthesis of safe control law under fixed interconnections.
When the interconnections change, the resilience indices and safe control law in \cite{al2022compositional} could not always guarantee safety property.
In this paper, we propose a resilience index and derive a system of linear inequalities that can applied to verify safety when interconnections change. 
When the inequalities are feasible, the interconnected system can satisfy the safety constraints.

\section{Problem Formulation}\label{sec:formulation}

We first define the notations that will be used throughout this paper.
Let $x\in\mathbb{R}^n$. 
We denote the $k$-th entry of $x$ as $[x]_k$.
A continuous function $\alpha:[-b,a)\rightarrow(-\infty,\infty)$ belongs to extended class $\mathcal{K}$ if it is strictly increasing and $\alpha(0)=0$ for some $a,b>0$.
Linear functions $\alpha(x)=zx$ defined over $[-b,a)$ are extended class $\mathcal{K}$ functions when $z>0$.

We consider a collection of subsystems $\{\mathcal{S}_i\}_{i\in\mathcal{N}}$, where $\mathcal{N}=\{1,\ldots,N\}$. Each subsystem $\mathcal{S}_i$ individually follows dynamics given as 
\begin{equation}\label{eq:subsystem dynamics}
    \mathcal{S}_i:\quad \Dot{x}_i = f_i(x_i) + g_i(x_i)u_i,
\end{equation}
where $x_i\in\mathcal{X}_i\subseteq\mathbb{R}^{n_i}$ is the state of subsystem $i$, and $u_i\in\mathcal{U}_i\subseteq\mathbb{R}^{p_i}$ is the external control input $u_i$ applied to subsystem $i$. 
We consider that the external control input $u_i$ will be chosen following a feedback control policy $\mu_i:\mathbb{R}^{n_i}\rightarrow\mathcal{U}_i$. 
Given the control policy $\mu_i$ and an initial state $x_i(0)$ at time $t=0$, we denote the trajectory of subsystem $\mathcal{S}_i$ as $x_i(t;x_i(0),\mu_i)$.
Functions $f_i:\mathbb{R}^{n_{i}}\rightarrow \mathbb{R}^{n_i}$ and $g_i:\mathbb{R}^{n_{i}}\rightarrow \mathbb{R}^{n_i\times p_i}$ are locally Lipschitz continuous. 

We consider that each subsystem $\mathcal{S}_i$ is required to satisfy a safety constraint defined over a set $\mathcal{C}_i=\{x_i:h_i(x_i)\geq 0\}$, i.e., $x_i(t;x_i(0),\mu_i)\in\mathcal{C}_i$ is required to hold for all $t\geq 0$.
Function $h_i:\mathbb{R}^{n_i}\rightarrow\mathbb{R}$ is continuously differentiable. 
We assume that the subsystem is initially safe, i.e., $x(0)\in\mathcal{C}_i$.
For each subsystem $\mathcal{S}_i$, we assume that we are given a control law $\mu_i$ such that $x_i(t;x_i(0),\mu_i)\in\mathcal{C}_i$ holds for all $t\geq 0$.
Such a safe control law $\mu_i$ can be synthesized by using approaches such as control barrier functions \cite{ames2016control}.

We assume that each subsystem can be faulty or compromised by an adversary.
When the subsystem is faulty or under attack, the safe control law $\mu_i$ becomes offline, and the control input received by subsystem $\mathcal{S}_i$ can be arbitrarily altered to some $\Tilde{u}_i\in\mathcal{U}_i$ that deviates from $\mu_i$.
To mitigate the persistence of faults and attacks, we consider that the subsystem recovers control law $\mu_i$ after the occurrence of faults and attacks leveraging fault/ attack detection and isolation techniques \cite{pasqualetti2013attack}.

To capture the fact that faults and attacks cause the control input to deviate from $\mu_i$ to arbitrary $\Tilde{u}_i$, we represent each subsystem $\mathcal{S}_i$ under faults and attacks as a hybrid system $\mathcal{H}_i=(\mathcal{X}_i,\mathcal{U}_i,\mathcal{L},\mathcal{Y}_i,\mathcal{Y}_i^{init},Inv_i,\mathcal{F}_i,\allowbreak \Sigma_i,\mathcal{G}_i)$, where
\begin{itemize}
    \item $\mathcal{X}_i\subseteq\mathbb{R}^{n_i}$ is the continuous state space of subsystem $\mathcal{S}_i$, and  $\mathcal{U}_i\subseteq\mathbb{R}^{p_i}$ is the set of admissible control inputs.
    \item $\mathcal{L}=\{\texttt{offline}, \texttt{online}\}$ is a set of discrete locations capturing whether control law $\mu_i$ is available (\texttt{online}) or not (\texttt{offline}). 
    \item $\mathcal{Y}_i=\mathcal{X}_i\times\mathcal{L}$ is the state space of hybrid system $H$, and $\mathcal{Y}_i^{init}\subseteq \mathcal{Y}_i$ is the set of initial states.
    \item $Inv_i:\mathcal{L}\rightarrow 2^{\mathcal{X}_i}$ is the invariant that maps from the set of locations to the power set of $\mathcal{X}_i$. That is, $Inv_i(l)\subseteq\mathcal{X}_i$ specifies the set of possible continuous states when the system is at location $l\in\mathcal{L}$.
    \item $\mathcal{F}_i$ is the set of vector fields. For each $F_i\in\mathcal{F}_i$ in the form of Eqn. \eqref{eq:subsystem dynamics}, the continuous system state evolves as $\dot{x}_i=F_i(x_i,u_i,l)$, where $F_i$ is jointly determined by the system dynamics and the availability of control law $\mu_i$, and $\dot{x}_i$ represents the time derivative of $x_i$.
    \item $\Sigma_i\subseteq \mathcal{Y}_i\times \mathcal{Y}_i$ is the set of transitions between the states of the hybrid system. A transition $\sigma_i=((x_i,l),(x_i',l'))$ models the state transition from $(x_i,l)$ to $(x_i',l')$.
\end{itemize}

In application such as power systems and vehicle platoons, multiple subsystems are interconnected as a network.
The interconnections introduce couplings among subsystems, leading to the following dynamics for each subsystem
\begin{multline}\label{eq:interconnected dynamics}
    \Dot{x}_i = f_i(x_i) + g_i(x_i)\mu_i(x_i) \\
    + \sum_{j\neq i}W_{ji}(x_j,x_i)- \sum_{j\neq i}W_{ij}(x_i,x_j),
\end{multline}
where $W_{ij}(x_i,x_j)$ captures the interconnection between subsystems $\mathcal{S}_i$ and $\mathcal{S}_j$.
Note that interconnections are not necessarily symmetric, i.e., $W_{ij}(x_i,x_j)$ and $W_{ji}(x_j,x_i)$ may not be identical.

We denote the state and joint control input of the interconnected system as $x=[x_1^\top,\ldots,x_N^\top]^\top$ and $u=[u_1^\top,\ldots,u_N^\top]^\top$, respectively.
The interconnected system is therefore of high dimension and nonlinear.
Furthermore, when each subsystem can possibly be compromised or faulty, the number of faults or attacks incurred by the interconnected system is combinatorial, making safety verification computationally intractable.
In this paper, we investigate the following problem.


\begin{problem} \label{prob:interconnection}
Suppose that we are given a collection of subsystems $\{\mathcal{S}_i\}_{i=1}^N$ and their safe control laws $\mu_i$ with respect to their individual safety set $\mathcal{C}_i$, where $i\in\mathcal{N}$. 
The subsystems, which are potentially subject to faults and attacks, are interconnected within a network and each of them follows dynamics given in Eqn. \eqref{eq:interconnected dynamics}.
The goal is to verify whether the interconnected system satisfies the set of safety constraints defined over $\mathcal{C}_i$ for all $i=1,\ldots,N$.
\end{problem}

\section{Proposed Resilience Index}\label{sec:RI}

In this section, we propose a compositional resilience index for each subsystem to verify safety.

\subsection{Definition of Resilience Index}\label{subsec:hybrid system}

We note that the discrete location set $\mathcal{L}$ is uniform to all subsystems, allowing us to develop a unified index to measure the resilience of any subsystem under faults and attacks. 
Our insight is as follows.
At location \texttt{offline}, the safe control law $\mu_i$ is unavailable.
To avoid violating the safety constraint, we require the subsystem to stay within a set $\mathcal{D}_i\subseteq\mathcal{C}_i$ so that $x_i\in\mathcal{C}_i$ for all time when the hybrid system is at location \texttt{offline}.
When the hybrid system transitions from \texttt{offline} to \texttt{online}, the control law $\mu_i$ becomes available. 
Thereafter, the subsystem starts to \emph{recover} from faults and attacks.
To recover the control law and mitigate potential faults and attacks in the future, we let the control law $\mu_i$ to remain available for at least $\phi_i\geq 0$ amount of time. 
Furthermore, the system is required to reach set $\mathcal{D}_i$ within $\phi_i$ so that safety constraint will not be violated when attacks or faults occur in the future. 
Such insight allows us to define the following resilience index to capture each subsystem's resilience under faults and attacks.

\begin{definition}[Resilience Index of a Subsystem]\label{def:resilience index}
Consider a subsystem $\mathcal{S}_i$ that uses a feedback control law $\mu_i$ and is under a safety constraint defined on set $\mathcal{C}_i=\{x_i:h_i(x_i)\geq 0\}$.
Let subsystem $\mathcal{S}_i$ be formulated as hybrid system $\mathcal{H}_i$ and $d_i,\eta_i\geq0$, $\tau_i,\phi_i>0$.
We say subsystem $\mathcal{S}_i$ is $(d_i,\tau_i,\phi_i,\eta_i)$-resilient if the following conditions hold
\begin{itemize}
    \item Set $\mathcal{D}_i$ defined as $\mathcal{D}_i=\{x_i:h_i(x_i)-d_i\geq 0\}$ is forward invariant if control law $\mu_i$ is used and the hybrid system $\mathcal{H}_i$ in at location \emph{\texttt{online}}.
    \item After reaching location \emph{\texttt{offline}}, hybrid system $\mathcal{H}_i$ remains at location \emph{\texttt{offline}} for at most $\tau_i$ amount of time before transition from \emph{\texttt{offline}} to \emph{\texttt{online}} occurs for any $x_i$.
    \item Following the transition from \emph{\texttt{offline}} to \emph{\texttt{online}}, hybrid system $\mathcal{H}_i$ remains at location \emph{\texttt{online}} for at least $\phi_i$ amount of time.
    \item 
    When $x_i$ is at the boundary of $\mathcal{D}_i$ and the hybrid system $\mathcal{H}_i$ in at location \emph{\texttt{online}}, the time derivative of $x_i$ is lower bounded by $\eta_i$.
    Furthermore, given any state $x_i'\in\mathcal{C}_i\setminus\mathcal{D}_i$, the continuous state $x_i$ reaches $\mathcal{D}_i$ within $\phi_i$ amount of time when $\mathcal{H}_i$ is at location \emph{\texttt{online}} and control law $\mu_i$ is used.
\end{itemize}
The quadruple $(d_i,\tau_i,\phi_i,\eta_i)$ is the resilience index of $\mathcal{S}_i$.
\end{definition}


In what follows, we derive a set of conditions to compute the resilience index of a subsystem.
\begin{proposition}\label{prop:resilience index}
Consider a subsystem in Eqn. \eqref{eq:subsystem dynamics} under attack and a safety set $\mathcal{C}_i$. 
Let $h_i^d(x_i) = h_i(x_i)-d_i$ and $\mathcal{D}_i=\{x_i:h_i^d(x_i)\geq 0\}\subseteq\mathcal{C}_i$.
Suppose $x_i(0)\in\mathcal{D}_i$. If there exist constants $d_i,\eta_i\geq0$, $\tau_i,\phi_i>0$, and an extended class $\mathcal{K}$ function $\alpha_i(\cdot)$ such that
\begin{subequations}\label{eq:index definition}
\begin{align}
    &\frac{\partial h_i^d}{\partial x_i}(x_i)[f_i(x_i)+g_i(x_i)u_i]\geq -\frac{d_i}{\tau_i},~\forall (x_i,u_i)\in\mathcal{C}_i\times\mathcal{U}_i\label{eq:index definition 1}\\
    &\frac{\partial h_i^d}{\partial x_i}(x_i)[f_i(x_i)+g_i(x_i)\mu_i(x_i)]\geq \frac{d_i}{\phi_i},~\forall x_i\in\mathcal{C}_i\setminus\mathcal{D}_i\label{eq:index definition 2}\\
    &\frac{\partial h_i^d}{\partial x_i}(x_i)[f_i(x_i)+g_i(x_i)\mu_i(x_i)]\geq-\alpha_i(h_i^d(x_i))+\eta_i,\nonumber\\
    &\quad\quad\quad\quad\quad\quad\quad\quad\quad\quad\quad\quad\quad\quad\quad\quad\quad~\forall x\in\mathcal{D}_i\label{eq:index definition 3}
\end{align}
\end{subequations}
then subsystem $\mathcal{S}_i$ is safe with respect to $\mathcal{C}_i$.
\end{proposition}
\begin{proof}
Without loss of generality, we assume that the subsystem is compromised by attack at some arbitrary time $t_0\geq 0$. 
For any $t\in[t_0,t_0+\tau_i]$ and control input $u_i\in\mathcal{U}_i$, we have $h_i(x_i(t)) = h_i(x_i(t_0)) + \int_{s=t_0}^{t}\Dot{h}_i\dd s\geq d_i - \frac{d_i}{\tau_i}(t-t_0)\geq 0$,
where $\Dot{h}_i$ represents the time derivative of function $h_i$, the equality holds by the definition of $h_i(x_i(t))$, the first inequality holds by Eqn. \eqref{eq:index definition 1} and the observation that $\Dot{h}_i^d=\Dot{h}_i$, and the last inequality holds by $t\in[t_0,t_0+\tau]$.
Therefore, hybrid system $\mathcal{H}_i$ satisfies $x_i(t;x_i(t_0),u_i)\in\mathcal{C}_i$ for all $u_i\in\mathcal{U}_i$ when $x_i(t_0)\in\mathcal{D}_i$, $\mathcal{H}_i$ is at location $l=\text{\texttt{offline}}$ for at most $\tau_i$ amount of time, and Eqn. \eqref{eq:index definition 1} holds.

We next consider some arbitrary time $t\in [t_0+\tau_i,t_0+\tau_i+\phi_i]$ when the hybrid system is at location \texttt{online}.
We show that $x_i(t;x_i(t_0+\tau_i),\mu_i)\in\mathcal{C}_i$ for all $t\in[t_0+\tau_i,t_0+\tau_i+\phi_i]$ when the hybrid system remains at location \texttt{online} and control law $\mu_i$ is used.
We further show that there exists $t'\in [t_0+\tau_i,t_0+\tau_i+\phi_i]$ such that $x_i(t'';x_i(t_0+\tau_i),\mu_i)\in\mathcal{D}_i$ for all $t''\in[t',t_0+\tau_i+\phi_i]$ when the hybrid system remains at location \texttt{online} and control law $\mu_i$ is used.

We prove $x_i(t;x_i(t_0+\tau_i),\mu_i)\in\mathcal{C}_i$ for all $t\in[t_0+\tau_i,t_0+\tau_i+\phi_i]$ when the hybrid system remains at location \texttt{online} and $\mu_i$ is used by contradiction.
Suppose that the subsystem leaves the safety set.
Since the trajectory of $\mathcal{S}_i$ is continuous, then there exists time $t\in[t_0+\tau_i,t_0+\tau_i+\phi_i]$ such that $h(x(t))=0$ and $\dot{h}(x(t))<0$.
If such $x(t)\in\mathcal{C}_i\setminus\mathcal{D}_i$, we then have contradiction to Eqn. \eqref{eq:index definition 2} since $\frac{d_i}{\phi_i}>0$ and thus $\dot{h}(x(t))>0$.
If such $x(t)\in\mathcal{D}_i$, we also have contradiction since Eqn. \eqref{eq:index definition 3} implies that $\dot{h}(x(t))>\eta_i\geq0$.
Therefore, we can claim that $x_i(t;x_i(t_0+\tau_i),\mu_i)\in\mathcal{C}_i$ for all $t\in[t_0+\tau_i,t_0+\tau_i+\phi_i]$ when the hybrid system remains at location \texttt{online} and $\mu_i$ is used.

We now prove that there exists $t'\in [t_0+\tau_i,t_0+\tau_i+\phi_i]$ such that $x_i(t'';x_i(t_0+\tau_i),\mu_i)\in\mathcal{D}_i$ for all $t''\in[t',t_0+\tau_i+\phi_i]$ when the hybrid system remains at location \texttt{online} and control law $\mu_i$ is used.
Suppose no such $t'\in [t_0+\tau_i,t_0+\tau_i+\phi_i]$ exists.
We then have that $x_i(t_0+\tau_i+\phi_i)\in\mathcal{C}_i\setminus\mathcal{D}_i$ and hence $h_i(x_i(t_0+\tau_i+\phi_i))<d_i$.
By Eqn. \eqref{eq:index definition 2}, we have $h_i(x_i(t_0+\tau_i+\phi_i)) = h_i(x_i(t_0+\tau_i)) + \int_{s=t_0+\tau_i}^{t}\Dot{h}_i\dd s\geq \frac{d_i}{\phi_i}(t-t_0-\tau_i)\geq d_i$, leading to contradiction. 
Therefore, such $t'$ must exist.
Finally, by Eqn. \eqref{eq:index definition 3} and \cite{ames2016control}, we have that if $x_i(t';x_i(t_0+\tau_i),\mu_i)\in\mathcal{D}_i$, then set $\mathcal{D}_i$ if forward invariant when the hybrid system remains at location \texttt{online} and control law $\mu_i$ is used, indicating that $x_i(t'';x_i(t_0+\tau_i),\mu_i)\in\mathcal{D}_i$ for all $t''\in[t',t_0+\tau_i+\phi_i]$.
\end{proof}

\subsection{Computation of Resilience Index}\label{sec:resilience index sos}

In the following, we formulate a sum-of-squares (SOS) optimization program to compute the resilience index for any subsystem $\mathcal{S}_i$.
Under certain assumptions on the dynamics \eqref{eq:subsystem dynamics}, function $h_i(x)$, and control input set $\mathcal{U}_i$, we formulate the SOS program by converting the conditions in Eqn. \eqref{eq:index definition} into SOS constraints.
We make the following assumption.
\begin{assumption}\label{assump:semi-algebraic}
For any subsystem $\mathcal{S}_i$, we assume that functions $f_i(x_i)$, $g_i(x_i)$, and $h_i(x_i)$ are polynomial in $x_i$. In addition, we assume that $\mathcal{U}_i=\prod_{k=1}^{p_i}[u_{k,min},u_{k,max}]$ with $u_{k,min}<u_{k,max}$. 
\end{assumption}

In the following, we present the set of SOS constraints. 
We show that any $d_i,\eta_i\geq 0$ and $\tau_i,\phi_i>0$ satisfying the SOS constraints constitute the resilience index of $\mathcal{S}_i$.
\begin{proposition}\label{prop:sos}
Assume that control law $\mu_i$ is polynomial in $x_i$. Suppose there exist $d_i,\eta_i\geq 0$, and $\tau_i,\phi_i> 0$ such that the following expressions are SOS:
\begin{subequations}\label{eq:sos}
\begin{align}
    &\frac{\partial h_i^d}{\partial x_i}(x_i)[f_i(x_i)+g_i(x_i)u_i] +d_i\theta_i - q(x_i,u_i)h_i^d(x_i)\nonumber\\
    &\quad\quad\quad-\sum_{k=1}^{p_i}(w_k(x_i,u_i)([u]_k-[u]_{k,min})\nonumber\\
    &\quad\quad\quad\quad\quad\quad\quad\quad+v_k(x_i,u_i)([u]_{k,max}-[u]_k)),\label{eq:sos 1} \\
    &\frac{\partial h_i^d}{\partial x_i}(x_i)[f_i(x_i)+g_i(x_i)\mu_i(x_i)]-d_i\beta_i\nonumber\\
    &\quad\quad\quad\quad\quad\quad\quad\quad- l(x_i)h_i^d(x) + m(x_i)h_i(x_i),\label{eq:sos 2}\\
    &\frac{\partial h_i^d}{\partial x_i}(x_i)[f_i(x_i)+g_i(x_i)\mu_i(x_i)]\nonumber\\
    &\quad\quad\quad\quad\quad\quad\quad\quad +\alpha(h_i^d(x)) -r(x_i) h_i(x_i)-\eta_i,\label{eq:sos 3} 
\end{align}
\end{subequations}
where $l(x_i),m(x_i),q(x_i,u_i),r(x_i)$ are SOS, and $w_k(x_i,u_i)$ as well as $v_k(x_i,u_i)$ are SOS for each $k=1,\ldots,p_i$. Then $d_i$, $\tau_i=\frac{1}{\theta_i}$, $\phi_i=\frac{1}{\beta_i}$, and $\eta_i$ satisfy Eqn. \eqref{eq:index definition}.
\end{proposition}
\begin{proof}
We prove that Eqn. \eqref{eq:sos 1} implies Eqn. \eqref{eq:index definition 1}. The other SOS constraints can be proved in a similar manner.
Consider $x_i\in\mathcal{C}_i$ and $[u]_{k,min}\leq [u]_k\leq [u]_{k,max}$ for all $k=1,\ldots,p_i$. 
We thus have that $h_i^d(x_i)\geq 0$ since $\mathcal{D}_i\subseteq\mathcal{C}_i$.
In addition, we have that $[u]_k-[u]_{k,min}\geq 0$, and $[u]_{k,max}-[u]_k\geq 0$. 
When expression \eqref{eq:sos 1} is SOS, $q(x_i,u_i)$, $w_k(x_i,u_i)$, and $v_k(x_i,u_i)$ are SOS for all $k=1,\ldots, p_i$, 
we have that the expression in Eqn. \eqref{eq:sos 1} is non-negative.
Therefore, if expression \eqref{eq:sos 1} is SOS and $\tau_i$ is chosen as $\tau_i=1/\theta_i$, then Eqn. \eqref{eq:index definition 1} holds.
\end{proof}



We observe that the SOS constraints derived in Proposition \ref{prop:sos} are bilinear (see the terms $d_i\theta_i$ and $d_i\beta_i$). Hence the resilience index cannot be readily computed by implementing these SOS constraints.
We overcome this challenge by developing an alternating optimization procedure, as shown in Algorithm \ref{algo:compute index}.
In Algorithm \ref{algo:compute index}, parameters $\tau_{max}$ and $\phi_{min}$ are the upper and lower bounds for $\tau_i$ and $\phi_i$, respectively.
If there exist no bound for $\tau_i$ and $\phi_i$, then parameters $\tau_{max}$ and $\phi_{min}$ can be set as infinity and zero, respectively.
  \begin{center}
  	\begin{algorithm}[!htp]
  		\caption{Algorithm for computing resilience index}
  		\label{algo:compute index}
  		\begin{algorithmic}[1]
  			\State \textbf{Input}: $\mathcal{S}_i$, $\tau_{max}$, $\phi_{min}$, and $\epsilon>0$
  			\State \textbf{Output:} $d_i$, $\tau_i$, $\phi_i$, $\eta_i$
  			\State \textbf{Initialization:} $d_i=0$
  		    \While{$d_i\leq \sup_{x_i}\{h_i(x_i)\}$}
            \State Solve for $\theta_i$, $\beta_i$, $\eta_i$ such that \eqref{eq:sos} is feasible with $d_i$ fixed.
            \If{Eqn. \eqref{eq:sos} is feasible, $\frac{1}{\theta_i} \leq \tau_{max}$, and $\frac{1}{\beta_i} \geq \tau_{min}$}
            \State \textbf{return} $d_i$, $\tau_i=\frac{1}{\theta_i}$, $\phi_i=\frac{1}{\beta_i}$, and $\eta_i$
            \Else
            \State $d_i=d_i+\epsilon$
            \EndIf
            \EndWhile
  		\end{algorithmic}
  	\end{algorithm}
  \end{center}
\section{Resilience Index After Interconnection}\label{sec:interconnection}

In this section, we consider a setting where multiple subsystems, with each being formulated by a hybrid system $\mathcal{H}_i$, are interconnected within a network.
Suppose that a collection of subsystems $\{\mathcal{S}_i\}_{i=1}^N$ are interconnected within a network.
In the network, each subsystem $\mathcal{S}_i$ follows the dynamics as given by Eqn. \eqref{eq:interconnected dynamics}.

We note that the interconnected system can be formulated as a hybrid system $\mathcal{H}$ as defined in Section \ref{subsec:hybrid system}.
In this case, the continuous state space is $\mathcal{X}=\prod_{i=1}^N\mathcal{X}_i$, and the discrete location is $\mathcal{L}=\{\text{\texttt{offline}},\text{\texttt{online}}\}^N$.
We observe that the continuous state is of dimension $n=\sum_{i=1}^Nn_i$. 
Furthermore, the transitions among the discrete locations are the Cartesian product of discrete transitions of all subsystems, which is combinatorial in nature to capture all possible faults and attacks that can be incurred by the subsystems. 
Therefore, safety verification over the hybrid system $\mathcal{H}$ is computationally intractable for large-scale interconnected systems.
In what follows, we derive how the resilience index of each subsystem changes due to interconnections.
We further show how our proposed resilience index can be applied to efficiently verify safety constraints of the interconnected system.

\subsection{Computation of Resilience Index After Interconnection}

In the following, we first characterize the behaviors of any subsystem $\mathcal{S}_j$ when being interconnected.
We define 
\begin{equation}\label{eq:delta def}
\delta_j=\inf_{x}\Big\{\frac{\partial h_j^d}{\partial x_j}\big[\sum_{i\neq j}W_{ij}(x_i,x_j) -\sum_{i\neq j}W_{ji}(x_j,x_i)\big]\Big\}.    
\end{equation}
We will show that the resilience index of a subsystem  $\mathcal{S}_j$ after being interconnected can be bounded using one of the following two sets of inequalities
\begin{align}
    &R_1:
    \begin{cases}
    0\leq d_j'\leq d_j\leq \sup_{x_j}\{h_j(x_j)\}\\
    -\frac{d_j'}{\tau_j'}\leq -\frac{d_j}{\tau_j}+\delta_j \\
    \frac{\phi_j}{d_j+\phi_j\delta_j}d_j'-\phi_j'\leq 0\\
    \eta_j'\leq \delta_j+\min\{\frac{d_j}{\phi_j},\eta_j+\inf_{x_j\in\mathcal{D}_j}\{\alpha_j(h_j(x_j)-d_j')
    \\
    \quad\quad\quad\quad\quad\quad\quad\quad-\alpha_j(h_j(x_j)-d_j)\}\}
    \end{cases}\label{eq:R1}\\
    &R_2:
    \begin{cases}
    0\leq d_j\leq d_j'\leq \sup_{x_j}\{h_j(x_j)\}\\
    -\frac{d_j'}{\tau_j'}\leq -\frac{d_j}{\tau_j}+\delta_j \\
    \frac{d_j'}{\phi_j'}\leq \delta_j+\min\{\frac{d_j}{\phi_j},\inf_{x_j\in\mathcal{D}_j}\{-\alpha_j(h_j^d(x_j))\}+\eta_j\}\\
    \eta_j'\leq \delta_j+\eta_j-\sup_{x_j\in\mathcal{D}_j}\{\alpha_j(h_j(x_j)-d_j')\\
    \quad\quad\quad\quad\quad\quad\quad\quad-\alpha_j(h_j(x_j)-d_j)\}\}
    \end{cases} \label{eq:R2}
\end{align}

We define $\mathcal{D}_j'=\{x_j:h_j(x_j)-d_j'\geq 0\}$. The inequalities in $R_1$ and $R_2$ specify sets $\mathcal{D}_j'$ differently.
The inequalities in $R_1$ specify that $\mathcal{D}_j'\supseteq \mathcal{D}_j$, whereas $R_2$ defines $\mathcal{D}_j''\subseteq\mathcal{D}_j$, which further leads to distinct behaviors when $x_j\in\mathcal{C}_j\setminus\mathcal{D}_j'$.
In what follows, we show how  the behavior of each subsystem following dynamics in Eqn. \eqref{eq:interconnected dynamics} can be characterized by the solutions to  $R_1$ or $R_2$.
This allows us to further verify the safety constraints for the interconnected system.

\begin{theorem}\label{thm:RI interconnect}
Consider that a collection of subsystems $\{\mathcal{S}_j\}_{j=1}^N$ are interconnected, and each $\mathcal{S}_j$ follows dynamics as given in Eqn. \eqref{eq:interconnected dynamics} for all $j=1,\ldots,N$.
We denote their resilience indices before being interconnected as $(d_k,\tau_k,\phi_k,\eta_k)$, where $k=1,\ldots,N$.
Define $\mathcal{D}_j'=\{x_j:h_j(x_j)-d_j'\geq 0\}$.
If parameters $d_j',\tau_j',\phi_j'$, and $\eta_j'$ render either $R_1$ or $R_2$ to be feasible, then the following conditions hold for $\mathcal{S}_j$ after being interconnected:
\begin{subequations}\label{eq:RI definition cascade}
\begin{align}
    &\frac{\partial h_j^d}{\partial x_j}(x_j)[f_j(x_j)+g_j(x_j)u_j + \sum_{i\neq j}W_{ij}(x_i,x_j)\label{eq:RI definition cascade 1}\\
    &\quad-\sum_{i\neq j}W_{ji}(x_j,x_i)]
    \geq -\frac{d_j'}{\tau_j'},~\forall (x_j,u_j)\in\mathcal{C}_j\times\mathcal{U}_j\nonumber\\
    &\frac{\partial h_j^d}{\partial x_j}(x_j)[f_j(x_j)+g_j(x_j)\mu_j(x_j) + \sum_{i\neq j}W_{ij}(x_i,x_j)\label{eq:RI definition cascade 2}\\
    &\quad-\sum_{i\neq j}W_{ji}(x_j,x_i)]\geq \frac{d_j'}{\phi_j'},~\forall x_j\in\mathcal{C}_j\setminus\mathcal{D}_j'\nonumber\\
    &\frac{\partial h_j^d}{\partial x_j}(x_j)[f_j(x_j)+g_j(x_j)\mu_j(x_j) + \sum_{i\neq j}W_{ij}(x_i,x_j)\label{eq:RI definition cascade 3}\\
    &\quad-\sum_{i\neq j}W_{ji}(x_j,x_i)]\geq-\alpha_j(h_j(x_j)-d_j')+\eta_j',~\forall x\in\mathcal{D}_j'\nonumber
\end{align}
\end{subequations}
\end{theorem}
\begin{proof}
We first verify that if $d_j',\tau_j',\phi_j'$, and $\eta_j'$ satisfy $R_1$, then Eqn. \eqref{eq:RI definition cascade} holds.
We denote $\dot{h}_j^d$ as 
\begin{multline*}
    \dot{h}_j^d = \frac{\partial h_j^d}{\partial x_j}(x_j)[f_j(x_j)+g_j(x_j)\mu_j(x_j) \\
    + \sum_{i\neq j}W_{ij}(x_i,x_j) -\sum_{i\neq j}W_{ji}(x_j,x_i)].
\end{multline*}
Suppose that $d_j',\tau_j',\phi_j'$, and $\eta_j'$ yield $R_1$ to be feasible. 
In this case, we have that $\mathcal{D}_j=\{x_j:h_j(x_j)-d_j\geq 0\}\subseteq \mathcal{D}_j'=\{x_j:h_j(x_j)-d_j'\geq 0\}$ due to $d_j'\leq d_j$. 
When $\mathcal{H}_j$ is at location \texttt{offline}, we have
\begin{align}
    \dot{h}_j^d\geq& \frac{\partial h_j^d}{\partial x_j}(x_j)[f_j(x_j)+g_j(x_j)u_j]+\delta_j\label{eq:cascade RI R1-1}\\
    \geq & -\frac{d_j}{\tau_j}+\delta_j
    \geq -\frac{d_j'}{\tau_j'},~\forall (x_j,u_j)\in\mathcal{C}_j\times\mathcal{U}_j\label{eq:cascade RI R1-3}
\end{align}
where inequality \eqref{eq:cascade RI R1-1} holds Eqn. \eqref{eq:delta def}, the second inequality holds by Eqn. \eqref{eq:index definition 1}, and the last inequality holds by the assumption that $R_1$ is feasible under $(d_j',\tau_j',\phi_j',\eta_j')$.

Consider the case where hybrid system $\mathcal{H}_j$ is at location \texttt{online} and control law $\mu_j(x_j)$ is available. We have
\begin{equation}\label{eq:cascade RI R1-6}
    \dot{h}_j^d\geq  \frac{d_j}{\phi_j}+\delta_j\geq  \frac{d_j'}{\phi_j'}
\end{equation}
holds for all $x_j\in\mathcal{C}_j\setminus\mathcal{D}_j'$, where the first inequality holds by Eqn. \eqref{eq:index definition 2} and \eqref{eq:delta def}, and the second inequality holds by $(\mathcal{C}_j\setminus\mathcal{D}_j')\subseteq (\mathcal{C}_j\setminus\mathcal{D}_j)$ given the feasibility of $R_1$.

We finally consider the case where $x_j\in\mathcal{D}_j'$ by dividing our discussion into two scenarios.
When $x_j\in\mathcal{D}_j'\setminus\mathcal{D}_j$, Eqn. \eqref{eq:cascade RI R1-6} yields $\dot{h}_j^d
    \geq  \frac{d_j}{\phi_j} + \delta_j \geq \eta_j'$,
where the last inequality holds by the feasibility of $R_1$.
When $x_j\in\mathcal{D}_j\subseteq\mathcal{D}_j'$, we have that
\begin{align}\label{eq:cascade RI R1-8}
    &\dot{h}_j^d
    \geq  -\alpha(h_j^d(x_j))+\eta_j + \delta_j\geq -\alpha(h_j(x_j)-d_j')+\eta_j' 
\end{align}
holds for all $x_j\in\mathcal{D}_j$, where the last inequality holds by the feasibility of $R_1$, i.e., $\eta_j'\leq \delta_j+\eta_j+\inf_{x_j\in\mathcal{D}_j}\{\alpha(h_j(x_j)-d_j')-\alpha(h_j(x_j)-d_j)\}$.

We next verify that if $d_j',\tau_j',\phi_j'$, and $\eta_j'$ satisfy $R_2$, then Eqn. \eqref{eq:RI definition cascade} holds. 
Note that in this case, $\mathcal{D}_j'=\{x_j:h_j(x_j)-d_j'\geq 0\}\subseteq \mathcal{D}_j=\{x_j:h_j(x_j)-d_j\geq 0\}$. 
When hybrid system $\mathcal{H}_j$ is at location \texttt{offline}, Eqn. \eqref{eq:RI definition cascade 1} can be derived using Eqn. \eqref{eq:cascade RI R1-1} and \eqref{eq:cascade RI R1-3} given that $d_j',\tau_j',\phi_j'$, and $\eta_j'$ satisfy $R_2$.
We next consider that hybrid system $\mathcal{H}_j$ is at location \texttt{online}.
We discuss two possible scenarios that can occur when $x_j\in\mathcal{C}_j\setminus\mathcal{D}_j'$. 
If $x_j\in \mathcal{C}_j\setminus\mathcal{D}_j$, we have
\begin{equation}\label{eq:cascade RI R2-1}
    \dot{h}_j^d
    \geq \frac{d_j}{\phi_j}+\delta_j
    \geq \frac{d_j'}{\phi_j'},~\forall x_j\in\mathcal{C}_j\setminus\mathcal{D}_j
\end{equation}
where the first inequality holds by Eqn. \eqref{eq:index definition 2} and the definition of $\delta_j$, and the second inequality holds by the feasibility of $R_2$.
If $x_j\in \mathcal{D}_j\setminus\mathcal{D}_j'$, we have that
\begin{equation}\label{eq:cascade RI R2-2}
    \dot{h}_j^d
    \geq -\alpha(h_j^d(x_j))+\delta_j
    \geq \frac{d_j'}{\phi_j'},~\forall x_j\in\mathcal{D}_j\setminus\mathcal{D}_j',
\end{equation}
where the first inequality holds by Eqn. \eqref{eq:index definition 3} and the definition of $\delta_j$, and the second inequality holds by the feasibility of $R_2$.
Combining Eqn. \eqref{eq:cascade RI R2-1} and \eqref{eq:cascade RI R2-2} yields Eqn. \eqref{eq:RI definition cascade 2}.

We finally consider that $x_j\in\mathcal{D}_j'$. We have that
\begin{multline*}\label{eq:cascade RI R2-3}
    \dot{h}_j^d
    \geq  -\alpha(h_j^d(x_j))+\eta_j + \delta_j\geq -\alpha(h_j(x_j)-d_j')+\eta_j'
\end{multline*}
holds for all $x_j\in\mathcal{D}_j'$, where the first inequality holds by Eqn. \eqref{eq:index definition 3} and the definition of $\delta_j$, and the second inequality holds by $\mathcal{D}_j'\subseteq\mathcal{D}_j$ along with the feasibility of $R_2$.

Combining the discussion above completes the proof.
\end{proof}

We observe that when function $\alpha_j$ is linear, computing the resilience indices after interconnection reduces to solving a linear system.
In the following, we show that given the resilience indices of $\mathcal{S}_j$ before it is interconnected along with its control law $\mu_i$, we can efficiently quantify how its resilience index changes due to interconnections by solving a set of inequalities given in Eqn. \eqref{eq:R1} and \eqref{eq:R2}.
\begin{theorem}\label{thm:safety interconnect}
Consider that a collection of subsystems $\{\mathcal{S}_j\}_{j=1}^N$ are interconnected, and each $\mathcal{S}_j$ follows dynamics as given in Eqn. \eqref{eq:interconnected dynamics} for all $j=1,\ldots,N$.
We denote their resilience indices before being interconnected as $(d_k,\tau_k,\phi_k,\eta_k)$, where $k=1,\ldots,N$. 
If parameters $d_j',\eta_j'\geq 0$ and $\tau_j',\phi_j'>0$ satisfy either $R_1$ in Eqn. \eqref{eq:R1} or $R_2$ in Eqn. \eqref{eq:R2}, then $\mathcal{S}_j$ is $(d_j',\tau_j',\phi_j',\eta_j')$-resilient under dynamics \eqref{eq:interconnected dynamics}.
Furthermore, $\mathcal{S}_j$ is safe with respect to $\mathcal{C}_j$ after being interconnected within the network.
\end{theorem}
\begin{proof}
    When parameters $d_j',\tau_j',\phi_j'$, and $\eta_j'$ render either $R_1$ or $R_2$ to be feasible, we have that Eqn. \eqref{eq:RI definition cascade} holds by using Theorem \ref{thm:RI interconnect}.
    By using Proposition \ref{prop:resilience index} and Definition \ref{def:resilience index}, we have that when $d_j',\eta_j'\geq 0$ and $\tau_j',\phi_j'>0$, subsystem $\mathcal{S}_j$ is $(d_j',\tau_j',\phi_j',\eta_j')$-resilient and safe with respect to $\mathcal{C}_j$.
\end{proof}

Computing $\delta_j$ for each subsystem $\mathcal{S}_j$ requires to solve an optimization problem over joint system state $x\in\mathbb{R}^n$, where $n=\sum_{i=1}^Nn_i$. 
To alleviate the computations in high dimensional state space, we approximate parameter $\delta_j$ as 
\begin{equation*}
    \delta_j\geq \sum_{i\neq j}\inf_{(x_i,x_j)\in\mathcal{C}_i\times\mathcal{C}_j}\{W_{ij}(x_i,x_j)-W_{ji}(x_j,x_i)\}:=\Tilde{\delta}_j.
\end{equation*}
By replacing $\delta_j$ with $\Tilde{\delta}_j$ in Eqn. \eqref{eq:RI definition cascade}, we can apply similar approach to show that Theorem \ref{thm:RI interconnect} and \ref{thm:safety interconnect} still hold.

\subsection{Feasibility of Resilience Index for Interconnected System}

Consider an interconnected system consisting of $N$ subsystems. 
We need to determine whether the inequalities in $R_1$ or $R_2$ need to be solved to apply Theorem \ref{thm:RI interconnect} and \ref{thm:safety interconnect} for safety verification of the interconnected system.
One approach is to combine the inequalities in $R_1$ and $R_2$ by using a set of mixed integer constraints and big M-method \cite{boyd2004convex}, where the integer variable $y\in\{0,1\}$ models whether $R_1$ or $R_2$ is solved.
In this subsection, we show that we can determine whether $R_1$ or $R_2$ is feasible given the value of $\delta_j$, and hence avoid solving the mixed integer program.

\begin{theorem}\label{thm:feasibility}
  Consider a subsystem $\mathcal{S}_j$ whose resilience index is given as $(d_j,\tau_j,\phi_j,\eta_j)$ before being interconnected.
  Assume that $\alpha_j(h_j(x_j))=zh_j(x_j)$ for some coefficient $z>0$.
  If $\delta_j$ satisfies 
  \begin{equation}\label{eq:R1 feasibility}
      \delta_j\geq \max\{-\frac{d_j}{\phi_j},-\eta_j-zd_j\},
  \end{equation}
  then there exist $d_j',\eta_j'\geq 0$ and $\tau_j',\phi_j'>0$ such that the inequalities in $R_1$ are satisfied.
  If $\delta_j$ satisfies 
  \begin{equation}\label{eq:R2 feasibility}
      \delta_j\geq \max\{-\frac{d_j}{\phi_j},-\eta_j+z(\sup_{x_j}\{h_j(x_j)\}-d_j)\},
  \end{equation}
  then there exist $d_j',\eta_j'\geq 0$ and $\tau_j',\phi_j'>0$ such that the inequalities in $R_2$ are satisfied.
\end{theorem}
\begin{proof}
    Suppose that Eqn. \eqref{eq:R1 feasibility} holds. 
    We rewrite $R_1$ in the matrix form $A_jv_j\leq b_j$, where $v_j=[d_j',\tau_j',\phi_j',\eta_j']^\top$, 
    \begin{equation*}
        A = \begin{bmatrix}
        1 &0 &0 &0\\
        1 &-\frac{d_j}{\phi_j}+\delta_j &0 &0\\
        z &0 &0 &1\\
        -1 &0 &\frac{d_j}{\tau_j}-\delta_j &0\\
        0 &1 &1 &1
        \end{bmatrix}, 
        b = \begin{bmatrix}
        d_j\\
        0\\
        \delta_j+\eta_j+zd_j\\
        0\\
        \delta_j+\frac{d_j}{\phi_j}
        \end{bmatrix}
    \end{equation*}
    When Eqn. \eqref{eq:R1 feasibility} holds, we have that there exists no $r_j\geq 0$ such that $b_j^\top r_j<0$.
    Using Farkas' Lemma \cite{fang1993linear}, there must exist some $v_j\geq 0$ such that $A_jv_j\leq b_j$, and thus satisfies $R_1$.
    Similar proof technique can be applied to show that when Eqn. \eqref{eq:R2 feasibility} holds, there must exist non-negative $(d_j',\tau_j',\phi_j',\eta_j')$ such that $R_2$ is satisfied.
    Noticing that $\phi_j',\eta_j'=0$ will make $A_j$ and $b_j$ ill-defined completes our proof.
\end{proof}

Using Theorem \ref{thm:feasibility}, we can decide whether we need to solve the inequalities given by $R_1$ or $R_2$ according to the value of $\delta_j$. 
Therefore, we mitigate the computational complexity by solving $N$ sets of inequalities.
By observing that $-\eta_j-zd_j\leq -\eta_j+z(\sup_{x_j}\{h_j(x_j)-d_j)\})$, we further have that if there exist non-negative $(d_j',\tau_j',\phi_j',\eta_j')$ such that $R_2$ is satisfied, then there must also exist some non-negative solution to the inequalities in $R_1$.
Finally, the sign of $\delta_j$ can be used to reason whether interconnections improve the resilience.

\begin{proposition}\label{prop:positive delta}
Consider a subsystem $\mathcal{S}_j$ whose resilience index is given as $(d_j,\tau_j,\phi_j,\eta_j)$ before being interconnected.
If $\delta_j\geq 0$, then there exists a resilience index $(d_j',\tau_j',\phi_j',\eta_j')$ such that the set of inequalities given by $R_1$ is feasible.
Furthermore, the interconnections improve the resilience of $\mathcal{S}_j$ in the sense that
    \begin{equation*}
        d_j'\leq d_j,\quad \tau_j'=\tau_j, \quad \phi_j'\leq \phi_j,\quad \eta_j'\geq \eta_j.
    \end{equation*}
\end{proposition}
\begin{proof}
    We prove the proposition by giving a choice of non-negative $(d_j',\tau_j',\phi_j',\eta_j')$ that satisfies $R_1$.
    We first note that $\tau_j'=\tau_j\geq 0$ is valid choice for parameter $\tau_j'$.
    In this case, if $\delta_j\geq 0$, we have that $d_j'$ can be chosen as $d_j'=d_j+\delta_j\geq d_j\geq 0$. 
    Given this choice of $d_j'$, we have that $\phi_j'=\frac{d_j'}{d_j+\phi_j\delta_j}\phi_j\geq 0$.
    Since $\alpha$ is an extended class $\mathcal{K}$ function, we have that $\alpha(h_j(x_j)-d_j')\geq \alpha(h_j(x_j)-d_j)$ for all $x_j$.
    Therefore, $\inf_{x_j\in\mathcal{D}_j}\{\alpha(h_j(x_j)-d_j')- \alpha(h_j(x_j)-d_j)\}\geq 0$.
    We can thus choose $\eta_j'$ as $\eta_j'=\min\{\frac{d_j}{\phi_j},\inf_{x_j\in\mathcal{D}_j}\{\alpha(h_j(x_j)-d_j')- \alpha(h_j(x_j)-d_j)\}\}
        +\delta_j\geq 0$.
    Hence, $(d_j',\tau_j',\phi_j',\eta_j')$ is a valid resilience index for $\mathcal{S}_j$ after interconnection.
\end{proof}

\section{Case Study}\label{sec:case study}

In this section, we demonstrate how the proposed resilience index can be used to analyze safety constraints of interconnected systems.


We consider two well-mixed, nonisothermal continuous stirred-tank reactors (CSTRs), denoted as $\mathcal{S}_1$ and $\mathcal{S}_2$.
We assume that three parallel elementary irreversible exothermic reactions of the form $A\xrightarrow{r_1} B$, $B\xrightarrow{r_2} E$, and $A\xrightarrow{r_3}Q$ occurs, where $A$ is the reactant species, $B$ is the desired product, and $E$ as well as $Q$ are undesired byproducts.
The states of $\mathcal{S}_1$ and $\mathcal{S}_2$ are denoted as $x_1=[T_1,c_1]^\top$ and $x_2=[T_2,c_2]^\top$, where $T_i$ and $c_i$ respectively represent temperature of the reactor and concentration of $\mathcal{S}_i$ with $i\in\{1,2\}$.
Each CSTR utilizes a jacket to remove or provide heat to the reactor to control the chemical reaction.

When the CSTRs are interconnected in series, their dynamics are given as
\begin{align}
    \dot{x}_1&=\begin{bmatrix}
    \frac{F_{e,1}}{V_1}(T_{0,1}-T_1) - \sum_{r=1}^3\frac{H_r}{\rho p}R_r(c_1,T_1)\\
    \frac{F_{e,1}}{V_i}(c_{0,1}-c_1) - \sum_{r=1}^3R_r(c_1,T_1)
    \end{bmatrix} + \begin{bmatrix}
        \frac{u_1}{\rho p V_1}\\
        0
    \end{bmatrix} \label{eq:chemical 1}\\
    \dot{x}_2&=\begin{bmatrix}
    \frac{F_{e,2}}{V_2}(T_{0,2}-T_2) - \sum_{r=1}^3\frac{H_r}{\rho p}R_r(c_2,T_2)\\
    \frac{F_{e,2}}{V_2}(c_{0,2}-c_2) - \sum_{r=1}^3R_r(c_2,T_2)
    \end{bmatrix} + \begin{bmatrix}
        \frac{u_1}{\rho p V_1}\\
        0
    \end{bmatrix}\nonumber\\
    &\quad\quad\quad + \frac{F_1}{V_2}\begin{bmatrix}
        T_1-T_2\\
        c_1-c_2
    \end{bmatrix}\label{eq:chemical 2}
\end{align}
where $F_{e,i}$ is the flow rate, $R_r(c_i,T_i)=k_i\exp(-E_i/lT_i)c_i$, $V_i$ is the volume, $u_i$ is the rate of heat input/ removal.
We follow the choices of process parameters given in \cite{el2005fault} and summarize them in Table \ref{tab:values}.
The set of admissible inputs is set as $|u_1|\leq 2.7\times 10^6$~KJ/hr and $|u_2|\leq 2.8\times 10^6$~KJ/hr.
In Eqn. \eqref{eq:chemical 2}, we have that
\begin{equation*}
    W_{12}(x_1,x_2) = \frac{F_1}{V_2}\begin{bmatrix}
        T_1-T_2\\
        c_1-c_2
    \end{bmatrix}.
\end{equation*}

\begin{figure}
    \centering
    \includegraphics[scale = 0.42]{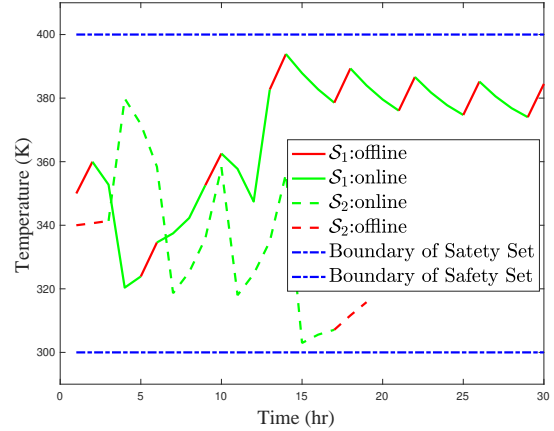}
    \caption{This figure presents the temperature in CSTRs $\mathcal{S}_1$ and $\mathcal{S}_2$ after they are interconnected in series. The temperature of $\mathcal{S}_1$ and $\mathcal{S}_2$ are represented in solid and dashed lines. The portion plotted in red represents that the CSTR is faulty or compromised, whereas the portion plotted in green denotes that the desired rate of heat input is available. }
    \label{fig:temperature}
\end{figure}

We assume that the safety constraints defined for CSTRs $\mathcal{S}_1$ and $\mathcal{S}_2$ are $\mathcal{C}_1=\{T_1:h_1(T_1)\geq 0\}$ and $\mathcal{C}_2=\{T_2:h_2(T_2)\geq 0\}$, where $h_i(T_i)=(T_i-300)(400-T_i)$ for all $i\in\{1,2\}$.
That is, the temperature in both CSTR needs to be within range $[300,400]$K.
Both CSTRs can be faulty or compromised, leading to manipulated rate of heat input $\Tilde{u}_i$, where $i\in\{1,2\}$.

\begin{table}[ht!]
\centering
\caption{This table presents the values of process parameters used in Eqn. \eqref{eq:chemical 1} and \eqref{eq:chemical 2}.}
\begin{tabular}{|c|c|} 
 \hline
 \textbf{Parameter Value} &  \textbf{Unit} \\ [0.5ex] 
 \hline
 $F_{e,1}=4.998$, $F_{e,2}=30.0$, $F_1=4.998$ & $m^3/hr$ \\
 \hline
 $V_1=1.0$, $V_2=3.0$ & $m^3$\\  
 \hline
 $T_{0,1}=300.0$, $T_{0,2}=300.0$ & K\\
 \hline
 $E_1=5.0$, $E_2=7.53$, $E_3=7.53$ & $\times 10^4$ KJ/kmol\\
 \hline
 $H_1=-5.0$, $H_2=-5.2$, $H_3=-5.4$ & $\times 10^4$ KJ/kmol\\
 \hline
 $\rho=1000.0$ & kg/$m^3$\\
 \hline
 $k_1=3.0\times 10^6$, $k_2=3.0\times 10^5$,  $k_3=3.0\times 10^5$ & hr$^{-1}$\\
 \hline
 $l=8.314$ &KJ/kmol$\cdot$K\\
 \hline
 $p=0.231$ &KJ/kg$\cdot$ K\\
 \hline 
 $c_{0,1}=4.0$, $c_{0,2}=2.0$  &kmol/$m^3$\\
 \hline 
\end{tabular}
\label{tab:values}
\end{table}

We first compute the resilience indices of CSTRs $\mathcal{S}_1$ and $\mathcal{S}_2$ when they are not interconnected.
Their resilience indices are given as $(d_1,\tau_1,\phi_1,\eta_1) = (2100,0.0146,0.308,0)$ and $(d_2,\tau_2,\phi_2,\eta_2) = (500,0.0292,0.0222,0)$.
Following dynamics Eqn. \eqref{eq:chemical 1} and \eqref{eq:chemical 2}, we have that the $\delta_{12}$ defined in Eqn. \eqref{eq:delta def} is negative, and hence we aim to solve $R_2$ to compute the resilience index of CSTR $\mathcal{S}_2$ after interconnection.
We have that $(d_2',\tau_2',\phi_2',\eta_2') = (800,0.0237,0.1368,0)$.
We observe that by fixing $\tau_2=\tau_2'$, $d_2'<d_2$, $\phi_2'<\phi_2$, and $\eta_2'>\eta_2$, which aligns with our result in Proposition \ref{prop:positive delta}.
We simulate the temperature in both CSTRs in Fig. \ref{fig:temperature}.
We plot the time period when the control input is compromised in red color, and the time period when the desired control law is online in green color.
We observe that the fault or attack could manipulate the temperature in both CSTRs by changing the rate of heat input $u_i$.
We further demonstrate that safety constraints defined on the temperature of $\mathcal{S}_1$ and $\mathcal{S}_2$ are met. 
We plot the boundaries of the safety set $\mathcal{C}_1$ and $\mathcal{C}_2$ using dash-dotted blue lines. 
We observe that $T_1$ and $T_2$ remain within $[300,400]$K for all time $t\geq 0$, and hence safety constraint is satisfied if we can find a feasible resilience index, which demonstrates Theorem \ref{thm:safety interconnect}.

\section{Conclusion}\label{sec:conclusion}

In this paper, we investigated the problem of efficient safety verification for large-scale interconnected systems under faults and attacks.
We developed a compositional resilience index for each subsystem to characterize its capability on tolerating faults and attacks without violating safety constraints.
We showed that if a subsystem possessed a resilience index, then it satisfies the given safety constraint regardless of the faults and attacks.
We formulated a sum-of-squares optimization program to compute the resilience index.
When the resilience index and a safe control law of a subsystem were given, we proved that the resilience index of the subsystem after being interconnected could be computed by solving a system of linear inequalities.
We further developed the sufficient conditions over the interconnections to guarantee the derived linear inequalities to be feasible.
We demonstrated the proposed approach using a case study on interconnected chemical reactors.

\bibliographystyle{IEEEtran}
\bibliography{MyBib}

\begin{thebibliography}{10}
\providecommand{\url}[1]{#1}
\csname url@rmstyle\endcsname
\providecommand{\newblock}{\relax}
\providecommand{\bibinfo}[2]{#2}
\providecommand\BIBentrySTDinterwordspacing{\spaceskip=0pt\relax}
\providecommand\BIBentryALTinterwordstretchfactor{4}
\providecommand\BIBentryALTinterwordspacing{\spaceskip=\fontdimen2\font plus
\BIBentryALTinterwordstretchfactor\fontdimen3\font minus
  \fontdimen4\font\relax}
\providecommand\BIBforeignlanguage[2]{{%
\expandafter\ifx\csname l@#1\endcsname\relax
\typeout{** WARNING: IEEEtran.bst: No hyphenation pattern has been}%
\typeout{** loaded for the language `#1'. Using the pattern for}%
\typeout{** the default language instead.}%
\else
\language=\csname l@#1\endcsname
\fi
#2}}

\bibitem{sullivan2017cyber}
J.~E. Sullivan and D.~Kamensky, ``How cyber-attacks in {U}kraine show the
  vulnerability of the {US} power grid,'' \emph{The Electricity Journal},
  vol.~30, no.~3, pp. 30--35, 2017.

\bibitem{el2005fault}
N.~H. El-Farra, A.~Gani, and P.~D. Christofides, ``Fault-tolerant control of
  process systems using communication networks,'' \emph{AIChE Journal},
  vol.~51, no.~6, pp. 1665--1682, 2005.

\bibitem{Jeep}
\BIBentryALTinterwordspacing
A.~Greenberg, ``Hackers remotely kill a {J}eep on the highway--with me in it,''
  2015. [Online]. Available:
  \url{https://www.wired.com/2015/07/hackers-remotely-kill-jeep-highway/}
\BIBentrySTDinterwordspacing

\bibitem{fan2020fast}
C.~Fan, K.~Miller, and S.~Mitra, ``Fast and guaranteed safe controller
  synthesis for nonlinear vehicle models,'' in \emph{International Conference
  on Computer Aided Verification}.\hskip 1em plus 0.5em minus 0.4em\relax
  Springer, 2020, pp. 629--652.

\bibitem{koscher2010experimental}
K.~Koscher, S.~Savage, F.~Roesner, S.~Patel, T.~Kohno, A.~Czeskis, D.~McCoy,
  B.~Kantor, D.~Anderson, H.~Shacham, and S.~Savage, ``Experimental security
  analysis of a modern automobile,'' in \emph{IEEE Symposium on Security and
  Privacy}.\hskip 1em plus 0.5em minus 0.4em\relax IEEE, 2010, pp. 447--462.

\bibitem{chen2020guaranteed}
Y.~Chen, A.~Singletary, and A.~D. Ames, ``Guaranteed obstacle avoidance for
  multi-robot operations with limited actuation: A control barrier function
  approach,'' \emph{IEEE Control Systems Letters}, vol.~5, no.~1, pp. 127--132,
  2020.

\bibitem{fijalkow2019decidability}
N.~Fijalkow, J.~Ouaknine, A.~Pouly, J.~Sousa-Pinto, and J.~Worrell, ``On the
  decidability of reachability in linear time-invariant systems,'' in
  \emph{22nd ACM International Conference on Hybrid Systems: Computation and
  Control}, 2019, pp. 77--86.

\bibitem{prajna2007framework}
S.~Prajna, A.~Jadbabaie, and G.~J. Pappas, ``A framework for worst-case and
  stochastic safety verification using barrier certificates,'' \emph{IEEE
  Transactions on Automatic Control}, vol.~52, no.~8, pp. 1415--1428, 2007.

\bibitem{niu2022verifying}
L.~Niu, D.~Sahabandu, A.~Clark, and P.~Radha, ``Verifying safety for resilient
  cyber-physical systems via reactive software restart,'' in \emph{ACM/IEEE
  13th International Conference on Cyber-Physical Systems (ICCPS)}.\hskip 1em
  plus 0.5em minus 0.4em\relax ACM/IEEE, 2022, pp. 104--115.

\bibitem{clark2021verification}
A.~Clark, ``Verification and synthesis of control barrier functions,'' in
  \emph{60th IEEE Conference on Decision and Control (CDC)}, 2021, pp.
  6105--6112.

\bibitem{manna2012temporal}
Z.~Manna and A.~Pnueli, \emph{Temporal Verification of Reactive Systems:
  Safety}.\hskip 1em plus 0.5em minus 0.4em\relax Springer Science \& Business
  Media, 2012.

\bibitem{mitchell2005time}
I.~M. Mitchell, A.~M. Bayen, and C.~J. Tomlin, ``A time-dependent
  {H}amilton-{J}acobi formulation of reachable sets for continuous dynamic
  games,'' \emph{IEEE Transactions on Automatic Control}, vol.~50, no.~7, pp.
  947--957, 2005.

\bibitem{bertsekas1971minimax}
D.~P. Bertsekas and I.~B. Rhodes, ``On the minimax reachability of target sets
  and target tubes,'' \emph{Automatica}, vol.~7, no.~2, pp. 233--247, 1971.

\bibitem{fisac2018general}
J.~F. Fisac, A.~K. Akametalu, M.~N. Zeilinger, S.~Kaynama, J.~Gillula, and
  C.~J. Tomlin, ``A general safety framework for learning-based control in
  uncertain robotic systems,'' \emph{IEEE Transactions on Automatic Control},
  vol.~64, no.~7, pp. 2737--2752, 2018.

\bibitem{sloth2012compositional}
C.~Sloth, G.~J. Pappas, and R.~Wisniewski, ``Compositional safety analysis
  using barrier certificates,'' in \emph{Proceedings of the 15th ACM
  International Conference on Hybrid Systems: Computation and Control}, 2012,
  pp. 15--24.

\bibitem{lyu2022small}
Z.~Lyu, X.~Xu, and Y.~Hong, ``Small-gain theorem for safety verification of
  interconnected systems,'' \emph{Automatica}, vol. 139, p. 110178, 2022.

\bibitem{coogan2014dissipativity}
S.~Coogan and M.~Arcak, ``A dissipativity approach to safety verification for
  interconnected systems,'' \emph{IEEE Transactions on Automatic Control},
  vol.~60, no.~6, pp. 1722--1727, 2014.

\bibitem{anand2021small}
M.~Anand, A.~Lavaei, and M.~Zamani, ``From small-gain theory to compositional
  construction of barrier certificates for large-scale stochastic systems,''
  \emph{arXiv preprint arXiv:2101.06916}, 2021.

\bibitem{clarke1997model}
E.~M. Clarke, ``Model checking,'' in \emph{International Conference on
  Foundations of Software Technology and Theoretical Computer Science}.\hskip
  1em plus 0.5em minus 0.4em\relax Springer, 1997, pp. 54--56.

\bibitem{ames2016control}
A.~D. Ames, X.~Xu, J.~W. Grizzle, and P.~Tabuada, ``Control barrier function
  based quadratic programs for safety critical systems,'' \emph{IEEE
  Transactions on Automatic Control}, vol.~62, no.~8, pp. 3861--3876, 2016.

\bibitem{xu2018constrained}
X.~Xu, ``Constrained control of input--output linearizable systems using
  control sharing barrier functions,'' \emph{Automatica}, vol.~87, pp.
  195--201, 2018.

\bibitem{breeden2021guaranteed}
J.~Breeden and D.~Panagou, ``Guaranteed safe spacecraft docking with control
  barrier functions,'' \emph{IEEE Control Systems Letters}, vol.~6, pp.
  2000--2005, 2021.

\bibitem{fawzi2014secure}
H.~Fawzi, P.~Tabuada, and S.~Diggavi, ``Secure estimation and control for
  cyber-physical systems under adversarial attacks,'' \emph{IEEE Transactions
  on Automatic control}, vol.~59, no.~6, pp. 1454--1467, 2014.

\bibitem{pajic2014robustness}
M.~Pajic, J.~Weimer, N.~Bezzo, P.~Tabuada, O.~Sokolsky, I.~Lee, and G.~J.
  Pappas, ``Robustness of attack-resilient state estimators,'' in
  \emph{ACM/IEEE International Conference on Cyber-Physical Systems (ICCPS)},
  2014, pp. 163--174.

\bibitem{sha2001using}
L.~Sha, ``Using simplicity to control complexity,'' \emph{IEEE Software},
  vol.~18, no.~4, pp. 20--28, 2001.

\bibitem{abdi2018guaranteed}
F.~Abdi, C.-Y. Chen, M.~Hasan, S.~Liu, S.~Mohan, and M.~Caccamo, ``Guaranteed
  physical security with restart-based design for cyber-physical systems,'' in
  \emph{2018 ACM/IEEE 9th International Conference on Cyber-Physical Systems
  (ICCPS)}.\hskip 1em plus 0.5em minus 0.4em\relax ACM/IEEE, 2018, pp. 10--21.

\bibitem{mertoguno2019physics}
J.~S. Mertoguno, R.~M. Craven, M.~S. Mickelson, and D.~P. Koller, ``A
  physics-based strategy for cyber resilience of {CPS},'' in \emph{Autonomous
  Systems: Sensors, Processing, and Security for Vehicles and Infrastructure
  2019}, vol. 11009.\hskip 1em plus 0.5em minus 0.4em\relax International
  Society for Optics and Photonics, 2019, p. 110090E.

\bibitem{yang2015fault}
H.~Yang, B.~Jiang, M.~Staroswiecki, and Y.~Zhang, ``Fault recoverability and
  fault tolerant control for a class of interconnected nonlinear systems,''
  \emph{Automatica}, vol.~54, pp. 49--55, 2015.

\bibitem{yang2020exponential}
H.~Yang, C.~Zhang, Z.~An, and B.~Jiang, ``Exponential small-gain theorem and
  fault tolerant safe control of interconnected nonlinear systems,''
  \emph{Automatica}, vol. 115, p. 108866, 2020.

\bibitem{al2022compositional}
A.~Al~Maruf, L.~Niu, A.~Clark, J.~S. Mertoguno, and R.~Poovendran, ``A
  compositional approach to safety-critical resilient control for systems with
  coupled dynamics,'' in \emph{61st Conference on Decision and Control
  (CDC)}.\hskip 1em plus 0.5em minus 0.4em\relax IEEE, 2022, pp. 910--917.

\bibitem{pasqualetti2013attack}
F.~Pasqualetti, F.~D{\"o}rfler, and F.~Bullo, ``Attack detection and
  identification in cyber-physical systems,'' \emph{IEEE Transactions on
  Automatic Control}, vol.~58, no.~11, pp. 2715--2729, 2013.

\bibitem{boyd2004convex}
S.~Boyd, S.~P. Boyd, and L.~Vandenberghe, \emph{Convex Optimization}.\hskip 1em
  plus 0.5em minus 0.4em\relax Cambridge university press, 2004.

\bibitem{fang1993linear}
S.-C. Fang and S.~Puthenpura, \emph{Linear Optimization and Extensions: Theory
  and Algorithms}.\hskip 1em plus 0.5em minus 0.4em\relax Prentice-Hall, Inc.,
  1993.

\end{thebibliography}

\end{document}